\documentclass[final,5p,times]{elsarticle}

\usepackage{tikz}
\usetikzlibrary{positioning}

\usepackage{caption}
\usepackage{csquotes}
\usepackage{amssymb}
\usepackage{amsmath}
\usepackage{amsthm}

\theoremstyle{plain}
\newtheorem{theorem}{Theorem}
\newtheorem{lemma}{Lemma}
\theoremstyle{remark}

\theoremstyle{definition}

\usepackage{float}
\usepackage{algorithm}
\usepackage{algpseudocode}
\usepackage[colorlinks=true,linkcolor=blue,citecolor=blue,urlcolor=blue]{hyperref}

\journal{Information Processing Letters}

\newcommand{\code}[1]{\texttt{\detokenize{#1}}}
\makeatletter
\newcommand{\breaktwo}{\if@twocolumn\\\fi}
\makeatother
\begin{document}

\begin{frontmatter}

\title{Multi-Way Co-Ranking: Index-Space Partitioning of Sorted Sequences Without Merge}

\author{
Amit Joshi \\
\href{mailto:amitjoshi2724@gmail.com}{\code{amitjoshi2724@gmail.com}} \\
\href{mailto:amit.joshiusa@gmail.com}{\code{amit.joshiusa@gmail.com}}
}
\makeatletter
\if@twocolumn
  \address{Independent Researcher\vspace{-5em}}
\else
  \address{Independent Researcher\vspace{-4em}}
\fi
\makeatother

\begin{abstract}
We present a merge-free algorithm for \emph{multi-way co-ranking}, the problem of computing \emph{cut indices} $i_1,\dots,i_m$ that partition each of the $m$ sorted sequences such that all prefix segments together contain exactly $K$ elements.
Our method extends Siebert and Träff's \cite{Siebert2025InPlaceMergesort} \cite{Siebert2014LoadBalancedMerge} two-list co-ranking to arbitrary $m$, maintaining per-sequence bounds that converge to a consistent global frontier without performing any multi-way merge or value-space search. Rather, we apply binary search to \emph{index-space}.
The algorithm runs in $\mathcal{O}(\log(\sum_t n_t)\,\log m)$ time and $\mathcal{O}(m)$ space, independent of $K$. 
We prove correctness via an exchange argument and discuss applications to distributed fractional knapsack, parallel merge partitioning, and multi-stream joins.
\end{abstract}


\end{frontmatter}

\section{Introduction}
\label{sec:intro}

\paragraph{Problem.}
Given $m$ sequences $L_1,\dots,L_m$, each sorted in non-decreasing order (duplicates allowed) by convention and a global target rank $K\in\{0,\dots,N\}$ with $N=\sum_t n_t$, the \emph{multi-way co-ranking} problem asks for cut indices $i_1,\dots,i_m$ such that:
\[
\sum_{t=1}^m i_t = K
\quad\text{and}\quad
\max_{t}\, \ell_t \;\le\; \min_{t}\, r_t,
\]
where $\ell_t = (i_t>0)\,?\,L_t[i_t-1]:-\infty$ and $r_t=(i_t<n_t)\,?\,L_t[i_t]:+\infty$ (sentinels at ends). Intuitively, the first $i_t$ items of $L_t$ are part of a global $K$-prefix (up to ties).

\paragraph{Contributions.}
We introduce a merge-free, \emph{index-space} algorithm that generalizes classic two-list co-ranking to arbitrary $m$:
\begin{itemize}
  \item We design an iterative \textbf{donor/receiver} greedy scheme driven by \emph{max-left} (largest $\ell_t$) and \emph{min-right} (smallest $r_t$) respectively and move mass between the two with the invariant that total mass always equals $K$. We store explicit upper and lower bounds for each list $Lb_t \leq i_t \leq Ub_t$.
  
  \item We show that \textbf{halving} the lower or upper range of the donor or receiver during each iteration yields fast convergence ($Lb_t = i_t = Ub_t$) in just

   \[
    \mathcal{O}\!\Big(\log(\!\sum_t n_t)\Big)\text{ iterations.}
  \]
  
  \item We use \textbf{indexed (addressable) heaps} to efficiently query $\ell_t$ and $r_t$ each iteration, and achieve a clean complexity bound independent of $K$:
  \[
    \mathcal{O}\!\Big(\log(\!\sum_t n_t)\cdot \log m\Big)\text{ time,}\qquad \mathcal{O}(m)\text{ space.}
  \]
  \item We give an \textbf{exchange-style correctness} proof: among all feasible infinitesimal transfers, moving mass from max-left to min-right achieves maximal progress on a natural imbalance potential, and any optimal correction sequence can be reordered so extremal transfers happen first without loss.
  \item We discuss \textbf{applications}: distributed fractional knapsack (multi-source), parallel merge partitioning, and multi-stream join partitioning.
\end{itemize}

\paragraph{Why index space?}
Many selection approaches reason in value space: merge the streams, maintain a heap of cursor values, or biselect a numeric threshold. Our method never merges or searches the value domain; it reasons \emph{only} on indices and the relative order that the sorted property guarantees. This aligns with co-ranking’s roots \cite{Siebert2014LoadBalancedMerge} and leads to a deterministic, $K$-independent complexity, just like in \cite{Siebert2025InPlaceMergesort}.

\section{Background and Related Work}
\label{sec:background}

\paragraph{Two-list co-ranking (partitioners for merge).}
Siebert and Traff's \cite{Siebert2014LoadBalancedMerge} \emph{co-ranking} for two sorted inputs (used to partition work in parallel merge) finds $(j,k)$ s.t.\ $j+k=i$ and the order constraints across the merge frontier hold. Efficient implementations do a binary-search-like \emph{halving} of the feasible interval, yielding $\mathcal{O}(\log n+\log m)$ time for the $i$-th output frontier. Siebert later used 2-way co-ranking work in \cite{Siebert2025InPlaceMergesort} to create an in-place merge-sort, but again among only 2 sorted sequences. Our work \emph{generalizes the partitioner} from 2-way to $m$-way, in index space and neatly describes the changes needed to do so.

\paragraph{Splitter-based parallel sorting (Siebert \& Wolf).}
Siebert and Wolf \cite{Siebert2011ExactSplitting} introduced an exact
splitter-finding method for scalable parallel sorting.
Their algorithm computes $p-1$ \emph{splitter values} that partition the global
key space so each process receives an equal share of elements,
achieving $O(\tfrac{n}{p}\log n + p\log^2 n)$ time.
This approach operates entirely in \emph{value space}: it searches for key thresholds
that balance partitions during redistribution.
By contrast, our multi-way co-ranking algorithm operates in
\emph{index space} on already-sorted sequences and directly outputs the
co-rank vector $(i_1,\dots,i_m)$ satisfying $\sum i_t=K$ and
$\max \ell_t\le \min r_t$, without any global value search.
Splitters compute thresholds; co-ranking computes frontier indices and partitions sequences into prefixes and suffixes independent of values.

\paragraph{Selection in sorted partitions (Frederickson--Johnson).}
A classical line of work due to Frederickson and Johnson \cite{FredericksonJohnson1980} studies \emph{selection} and \emph{ranking} over structured inputs such as the union of $m$ sorted lists (“sorted partitions”). Their algorithms are fundamentally \emph{value-space} procedures: they maintain numeric thresholds (candidate keys), repeatedly count how many elements are $\le$ a threshold across the lists (via one-sided searches), and narrow the feasible value interval until the $K$-th \emph{value} is identified. In the multi-list setting this yields bounds of the form $O\!\big(m + \sum_{t=1}^m \log n_t\big)$ for selecting the $K$-th smallest from $m$ sorted lists, with closely related results for ranking a given key \cite{FredericksonJohnson1980}. 
By contrast, our contribution is an \emph{index-space} primitive: we do not search the key domain or perform global ``count-$\le$'' operations. Instead, we directly evolve the per-sequence \emph{cut indices} $(i_1,\dots,i_m)$ satisfying $\sum_t i_t=K$ and the co-rank frontier condition, using only local comparisons at the current indices. Thus, while both approaches address ``$K$-th from multiple sorted sources''–type problems, Frederickson--Johnson output a \emph{value} via thresholding, whereas we output the full \emph{co-rank vector} without value bisection.

\paragraph{This paper’s positioning.}
We treat \emph{multi-way co-ranking} as a merge-free \emph{partitioning primitive} that is an efficient and simple generalization of 2-way co-ranking described in \cite{Siebert2014LoadBalancedMerge}. The novelty is methodological (efficient indexed heaps for fast updates, explicit upper and lower bounds) and conceptual (the dynamics of halving the feasible range, operating in index-space), with applications in fast queries in distributed environments, databases, and parallel partitioning and/or merging of data.

\section{Preliminaries and Notation}
\label{sec:prelim}

We have $m$ sorted sequences $L_t$ of lengths $n_t$, total $N=\sum_t n_t$, and a target rank $K\in[0,N]$. We maintain an index vector $i=(i_1,\dots,i_m)$ with feasibility \emph{bounds}
\[
Lb[t] \le i[t] \le Ub[t], \qquad \sum_t i[t] = K.
\]
Lower bounds $Lb$ are initialized at $0$. Upper bounds $Ub$ are initialized at $n_t$. 
We define left/right boundary values with sentinels on the ends:
\[
\ell_t = \begin{cases}
(-\infty, t) & i_t=0\\
(L_t[i_t-1], t) & \text{else}
\end{cases}
\qquad
r_t = \begin{cases}
(+\infty, t) & i_t=n_t\\
(L_t[i_t], t) & \text{else}
\end{cases}
\]
The co-rank condition is $\max_t \ell_t \le \min_t r_t$.
In cases of ties, we adopt a simple canonical convention:
lower list indices are treated as ``smaller,'' so this ordering applies consistently in both the min-heap and max-heap.
For the max-heap, this means that if two boundary values are equal,
the element from the later list index is considered ``larger'' and will be returned first.
This tie-breaking rule is purely conventional and ensures deterministic behavior;
if the specific ordering of equal elements is irrelevant to an application,
the tie-handling condition may be omitted altogether. If the co-rank condition is not satisfied, we select a donor $p$ and receiver $q$. We define the donor slack to be $i[p] - Lb[p]$ --- how much $p$ can donate before hitting its lower bound. Similarly, we define the receiver "slack" to be $Ub[q] - i[q]$ how much mass $q$ can receive before hitting its upper bound. We tighten at least one of these bounds on every iteration after the first one.

\section{Data Structure: Indexed (Addressable) Heaps}
\label{sec:indexedheap}

We maintain two indexed heaps over the fixed ID set $\{0,\dots,k-1\}$:
\begin{itemize}
  \item HL: a \emph{max-heap} on keys $(\ell_t, t)$, returning the donor $p$ with largest $\ell_t$; on elemental ties, larger $t$ first (or your chosen tie).
  \item HR: a \emph{min-heap} on keys $(r_t, t)$, returning the receiver $q$ with smallest $r_t$; on elemental ties, smaller $t$ first.
\end{itemize}
Each supports $\mathcal{O}(\log m)$ \texttt{update\_key(id,new\_key)} and $\mathcal{O}(1)$ \texttt{peek}. Because indices change only for at most two lists per round (donor and one receiver), each round incurs $O(\log m)$ heap work.

\begin{algorithm}[h]
\caption{Indexed heap (interface summary)}
\begin{algorithmic}[1]
\State \texttt{insert(id, key)} \Comment{$O(\log m)$}
\State \texttt{update\_key(id, newkey)} \Comment{$O(\log m)$ sift up/down}
\State \texttt{peek()} $\rightarrow$ (\textit{key}, \textit{id}) \Comment{$O(1)$}
\end{algorithmic}
\end{algorithm}

\section{Algorithm}
\label{sec:algorithm}

\paragraph{High-level idea.}
Start with any feasible $i$ (can water-fill the first however many lists) such that $\sum_{t = 1}^m i_t = K$. If $\max \ell \le \min r$ we are done. Otherwise, choose donor $p=\arg\max \ell$ and receiver $q=\arg\min r$. Consider the smaller of the two slacks (for example, if the donor has less room than the receiver, then select that), halve that slack, and transfer that mass from the donor to the receiver. This ensures that both the donor and receiver stay in their respective feasible zones. Update bounds and heap keys and repeat. Because we always transfer from the worst left boundary to the best right boundary, a simple potential function decreases monotonically. We describe our algorithm in detail in Algorithm \ref{alg:mwcorank}.
\begin{algorithm}[h]
\caption{Multi-Way Co-Ranking}
\label{alg:mwcorank}
\begin{algorithmic}[1]
\Require Sorted sequences $L_1,\dots,L_m$ (ascending), lengths $n_t=|L_t|$, target rank $K\in[0,\sum_t n_t]$.
\Ensure Indices $i_1,\dots,i_m$ with $\sum_{t=1}^m i_t=K$ and $\max_t \ell_t \le \min_t r_t$.
\Statex

\State \textbf{Bounds:} $L_b[t]\gets 0$ for all $t\in\{1,\dots,m\}$; \quad $U_b[t]\gets n_t$.
\State \textbf{Init feasible }$i$: set $i_t\gets L_b[t]$; \quad $\texttt{need}\gets K$.
\For{$t=1$ \textbf{to} $m$} \Comment{water-fill up to $K$ using capacities $(U_b[t]-i_t)$}
  \State $\texttt{cap}\gets U_b[t]-i_t$;\quad $\texttt{take}\gets \min\{\texttt{cap},\texttt{need}\}$;
  \State $i_t\gets i_t+\texttt{take}$;\quad $\texttt{need}\gets \texttt{need}-\texttt{take}$;\quad \textbf{if }$\texttt{need}=0$ \textbf{then break}
\EndFor
\Statex

\State \textbf{Boundary access:}
\State $\ell_t \gets \begin{cases}-\infty & \text{if } i_t=0\\ L_t[i_t-1] & \text{otherwise}\end{cases}
\quad\ $ and $\quad
r_t \gets \begin{cases}+\infty & \text{if } i_t=n_t\\ L_t[i_t] & \text{otherwise.}\end{cases}$
\Statex

\State \textbf{Heaps:} Build an \emph{indexed} max-heap HL keyed by $\ell_t$ (ties $\to$ larger $t$ first), and an indexed min-heap HR keyed by $r_t$ (ties $\to$ smaller $t$ first).
\For{$t=1$ \textbf{to} $m$} \State HL.\texttt{insert}$(t,\ell_t)$;\; HR.\texttt{insert}$(t,r_t)$ \EndFor
\Statex

\For{$\text{round}=1,2,\dots$}
  \State $(x,p)\gets$ HL.\texttt{peek()} \Comment{donor candidate: max-left}
  \State $(y,q)\gets$ HR.\texttt{peek()} \Comment{receiver candidate: min-right}
  \If{$x<y$ \textbf{or} $(x=y \ \land\ p\le q)$} \Comment{authoritative stop condition}
     \State \textbf{return } $i_1,\dots,i_m$
  \EndIf
  \Statex
  \State \textbf{Halving magnitudes on both sides:}
  \State $\Delta_p \gets \left\lceil \dfrac{i_p - L_b[p]}{2}\right\rceil$ \Comment{donor slack halved}
  \State $\Delta_q \gets \left\lceil \dfrac{U_b[q] - i_q}{2}\right\rceil$ \Comment{receiver headroom halved}
  \State $\Delta \gets \min\{\Delta_p,\ \Delta_q\}$ \Comment{move at most what both sides can support}
  \If{$\Delta=0$}
     \State \textbf{(degenerate case)} \textbf{break} \Comment{cannot progress; should not occur if not done}
  \EndIf
  \Statex
  \State \textbf{Tighten explicit bounds at the endpoints of the move:}
  \State $U_b[p]\gets i_p$ \Comment{donor cannot exceed its current index this round}
  \State $L_b[q]\gets i_q$ \Comment{receiver cannot go below its current index this round}
  \Statex
  \State \textbf{Commit the balanced transfer and refresh heap keys:}
  \State $i_p\gets i_p - \Delta$;\quad $i_q\gets i_q + \Delta$
  \State Update $\ell_p,r_p$ and $\ell_q,r_q$; \ \ HL.\texttt{update\_key}$(p,\ell_p)$, HR.\texttt{update\_key}$(p,r_p)$
  \State HL.\texttt{update\_key}$(q,\ell_q)$, HR.\texttt{update\_key}$(q,r_q)$
\EndFor
\end{algorithmic}
\end{algorithm}

\section{Correctness}
\label{sec:correctness}

We sketch a concise correctness argument.  
Let $\Phi(i) = \max_t \ell_t - \min_t r_t$ denote the imbalance potential; the algorithm terminates when $\Phi \le 0$.

\begin{lemma}[Local extremal optimality]
If $\Phi(i) > 0$, let $p = \arg\max_t \ell_t$ and $q = \arg\min_t r_t$.
Among all feasible infinitesimal transfers of index mass that preserve $\sum_t i_t = K$,
the pair $(p,q)$ achieves the greatest non-increasing change in $\Phi$.
Any other pair $(u,v)$ can produce at most the same or a weaker decrease in $\Phi$.
\end{lemma}

\begin{proof}[Proof sketch]
Decreasing $i_p$ reduces $\ell_p$ (the local maximum on the left boundary),
and increasing $i_q$ raises $r_q$ (the local minimum on the right boundary).
Since $\ell_p \ge \ell_u$ and $r_q \le r_v$ for all lists $u,v$,
the extremal pair $(p,q)$ produces the steepest possible reduction of the
gap $\max \ell - \min r$.  
Tie handling through the heap comparator (lower index first on the right, higher index first on the left) deterministically resolves symmetry but does not affect monotonicity of $\Phi$.
\end{proof}

\begin{lemma}[Exchange of transfer order]
Any finite sequence of feasible transfers that decreases $\Phi$
can be reordered so that all extremal $(p,q)$ transfers occur before any
non-extremal transfer, without worsening $\Phi$ at any intermediate step.
\end{lemma}

\begin{proof}[Proof sketch]
Each transfer modifies only two indices $(i_u,i_v)$ and can affect $\Phi$
only through the associated $\ell_u$ or $r_v$.  
By the previous lemma, the extremal pair has dominant effect:
executing it earlier cannot increase $\Phi$.  
Commuting a non-extremal move past an extremal one therefore preserves or improves the potential.  
Repeatedly commuting yields a sequence in which all extremal moves occur first and $\Phi$ is weakly smaller at every stage.
\end{proof}

\begin{theorem}[Convergence and validity]
Algorithm~\ref{alg:mwcorank} terminates with a valid co-rank vector
$(i_1,\dots,i_m)$ satisfying $\sum_t i_t = K$ and $\max_t \ell_t \le \min_t r_t$.
\end{theorem}

\begin{proof}[Proof sketch]
Feasibility ($L_b \le i \le U_b$ and $\sum i_t = K$) is maintained by construction:
each update transfers exactly $\Delta$ units of index mass from a donor $p$ to a receiver $q$
and tightens their respective bounds ($U_b[p] \gets i_p$, $L_b[q] \gets i_q$).  
At each round the potential $\Phi$ strictly decreases,
because the move reduces either $\ell_p$ or $r_q$ (whichever side was limiting),
or else a bound becomes saturated and is excluded from further motion.  
Since both $L_b$ and $U_b$ are monotone and integer-valued,
the process admits only finitely many distinct states and must terminate.  
By the exchange lemma, the greedy extremal choice made in each round is at least as effective as any alternative sequence.  
Upon termination $\Phi \le 0$, which by definition implies $\max_t \ell_t \le \min_t r_t$.
\end{proof}

\section{Complexity Analysis}
\label{sec:complexity}

Let $n_t=|L_t|$ and $N=\sum_t n_t$. We outline the cost components.

\paragraph{Rounds.}
In each round, at least one side---either the donor or the receiver---is halved,
whichever has the smaller available distance to its bound.
If the donor $p$ has the smaller distance, its slack 
$s_p = i_p - Lb[p]$ is halved
($i_p \gets i_p - \lceil s_p/2\rceil$ and $Ub[p] \gets i_p$);
otherwise, the receiver $q$ has its slack/headroom 
$h_q = Ub[q] - i_q$ halved 
($i_q \gets i_q + \lceil h_q/2\rceil$ and $Lb[q] \gets i_q$).
Each list’s distance to feasibility, 
$U_b[t] - L_b[t]$, can therefore shrink only 
$O(\log n_t)$ times before stabilizing.
Aggregating across all $m$ lists and observing that the global potential
$\Phi = \sum_t (U_b[t] - L_b[t])$ decreases geometrically, the total number of
rounds is
\[
T = \mathcal{O}\!\Big(\log\!\Big(\sum_{t=1}^m n_t\Big)\Big).
\]
Since each round performs a constant number of indexed heap operations
($O(\log m)$ time), the overall runtime is \break
$\mathcal{O}(\log(\sum_t n_t)\,\log m)$,
with linear (in the number of sequences) $\mathcal{O}(m)$ space.

\paragraph{Total.}
Combining, the time is
\[
\boxed{\ \mathcal{O}\!\Big(\log(\!\sum\nolimits_t n_t)\cdot \log k\Big)\ },
\qquad
\text{space } \boxed{\ \mathcal{O}(k)\ }.
\]

\section{Applications}
\label{sec:apps}

\paragraph{Distributed fractional knapsack (multi-source).}
The fractional knapsack problem is a variation of the classic knapsack (take some jewels such that their total value is maximized yet their cumulative weight fits in a given capacity) problem where taking a fraction of a jewel in your knapsack is allowed. This admits a greedy solution, where we sort jewels in decreasing order by density (value per weight), and select the highest density jewels. When taking the next jewel would put our knapsack over capacity, we cut that jewel so that its weight fits with the others under a total capacity, and add the proportional value to our score.

When a fractional knapsack must draw from $m$ sources (shards) already individually sorted (in non-increasing order) by density, one can compute the global $K$-prefix split (at a suitable discretization of mass) without merging sources. We can binary search on the largest $K$, such that the prefix segments have cumulative weight within capcity. For these efficient prefix sum queries, we can use an augmented (with sum) self-balancing BST or perhaps a similar advanced Linked List that allows for $O(\log(n))$ queries, insertions, and deletions if each source is dynamic.

\paragraph{Parallel $m$-way merge partitioning.}
Assign processors disjoint prefixes without performing a preliminary merge. The co-rank vector determines exact hand-offs; processors then perform local merges on their ranges only. This would essentially be Siebert and Wolfs' work \cite{Siebert2011ExactSplitting} with a different splitting mechanism -- one that operates in \emph{index-space} rather than binary searching on value splitters.

\paragraph{Multi-stream join partitioning (Databases).}
In database or stream-processing pipelines, partitioning the join frontier at a global rank is natural; our method yields the per-stream cursors consistent with the global prefix.

\section{Conclusion}
\label{sec:conclusion}

We presented an index-space, merge-free algorithm for \emph{multi-way co-ranking} that generalizes two-way co-ranking to arbitrary $m$-ways. The method uses explicit lower and upper bounds for each sequence, halving of feasible ranges, and an invariant that the cut indices for each partition add up to $K$, yielding a clean $\mathcal{O}(\log(\sum n_t)\cdot \log m)$ time and $\mathcal{O}(m)$ space bound, independent of $K$ (the global rank being queried). The approach is simple to implement, proof-friendly, and directly applicable to distributed optimization and parallel partitioning and merge tasks. Finally, as a thanks for reading, we release our implementation code.\footnote{\url{https://github.com/amitjoshi2724/multi-way-co-ranking}}

\bibliographystyle{plainnat}
\bibliography{references}
\end{document}